%% file: main.tex
\documentclass[a4paper,11pt]{article}

\usepackage[T1]{fontenc}
\usepackage{graphicx}
\usepackage{amsmath}
\usepackage{amssymb}
\usepackage{amsthm}
\usepackage{amsfonts}
\usepackage{comment}
\usepackage{xspace}
\usepackage{fullpage}

\usepackage[linesnumbered,noend,ruled,vlined]{algorithm2e}
\usepackage[noend]{algpseudocode}

\SetKwInput{KwInput}{Input}     
\SetKwInput{KwOutput}{Output}   
\SetKw{KwTo}{to}
\SetKwFor{RepTimes}{repeat}{times}{end}
\SetKwProg{Procedure}{Procedure}{}{}

\usepackage{thm-restate}

\usepackage{hyperref}
\usepackage{here}

\usepackage{cite}
\usepackage{enumerate}

\usepackage{color}

\newtheorem{theorem}{Theorem}

\newtheorem{lemma}{Lemma}
\newtheorem{claim}{Claim}
\newtheorem{definition}{Definition}

\newtheorem{corollary}{Corollary}

\newcommand{\coreset}{deterministic reachable kernel\xspace}
\newcommand{\weakcoreset}{reachable kernel\xspace}
\newcommand{\gkgood}{$g(k)$-coverable\xspace}
\newcommand{\gikgood}{$g_i(k)$-coverable\xspace}
\newcommand{\gkgoodf}[1]{#1-coverable\xspace}

\usepackage{xspace}
\newcommand{\I}{\mathcal{I}}
\newcommand{\M}{\mathcal{M}}
\newcommand{\F}{\mathcal{F}}

\global\long\def\span{\mathrm{span}}%
\global\long\def\eps{\varepsilon}%

\begin{document}
\title{Polynomial Kernels with Reachability\\ 
for Weighted $d$-Matroid Intersection\thanks{This work is partially supported by JSPS KAKENHI Grant Numbers 
JP22H05001, 
JP23K21646, 
JP24K02901, 
JP24KJ1494, 
Japan, and JST ERATO Grant Number JPMJER2301, Japan.}}
\author{Chien-Chung Huang\thanks{CNRS, DIENS, PSL, France, \texttt{Chien-Chung.Huang@ens.fr}}
\and 
Naonori Kakimura\thanks{Keio Unversity, Japan,  \texttt{kakimura@math.keio.ac.jp}}
\and 
Yusuke Kobayashi\thanks{Kyoto University, Japan, \texttt{yusuke@kurims.kyoto-u.ac.jp}} 
\and 
Tatsuya Terao\thanks{Kyoto University, Japan, \texttt{ttatsuya@kurims.kyoto-u.ac.jp}}
}
\date{}
\maketitle
\begin{abstract}
This paper studies randomized polynomial kernelization for the weighted $d$-matroid intersection problem.
While the problem is known to have a kernel of size $O(d^{(k - 1)d})$ where $k$ is the solution size,  the existence of a polynomial kernel is not known, except for the cases when either all the given matroids are partition matroids~(i.e., the $d$-dimensional matching problem) or all the given matroids are linearly representable.
The main contribution of this paper is to develop a new kernelization technique for handling general matroids.
We first show that the weighted $d$-matroid intersection problem admits a polynomial kernel when one matroid is arbitrary and the other $d-1$ matroids are partition matroids.
Interestingly, the obtained kernel has size $\tilde{O}(k^d)$, which matches the optimal bound~(up to logarithmic factors) for the $d$-dimensional matching problem.
This approach can be adapted to the case when $d-1$ matroids in the input belong to a more general class of matroids, including graphic, cographic, and transversal matroids.
We also show that the problem has a kernel of pseudo-polynomial size when given $d-1$ matroids are laminar.
Our technique finds a kernel such that any feasible solution of a given instance can reach a better solution in the kernel, which is sufficiently versatile to allow us to design parameterized streaming algorithms and faster EPTASs.
\end{abstract}

\input{introduction.tex}

\input{preliminaries.tex}

\input{partition.tex}

\input{transversal.tex}

\input{laminar.tex}

\input{matroidalmatching.tex}

\input{deterministic.tex}


\appendix

\input{matroidclass.tex}

\end{document}

%% file: introduction.tex
\section{Introduction} \label{sec:introduction}

In the \textit{weighted $d$-matroid intersection} problem, we are given $d$ matroids on a finite ground set $E$ and a weight function $w: E\to \mathbb{R}_{\ge 0}$.
The goal of the problem is to find a maximum-weight common independent set, that is, we aim to find a subset $X$ of $E$  that maximizes $\sum_{e\in X}w(e)$ subject to $X$ is an independent set for every matroid.
The weighted $d$-matroid intersection problem is one of the most fundamental problems in combinatorial optimization, as it includes the bipartite matching problem, the arborescence problem~(when $d=2$), and the $d$-dimensional matching problem as a special case.
The problem can be solved in polynomial time when $d\leq 2$~\cite{edmonds1968matroid, edmonds1979matroid, lawler1975matroid, aigner1971matching, iri1976algorithm}, while it is NP-hard when $d\geq 3$~(see~\cite{garey1979}).

In this paper, we study kernelization of the weighted $d$-matroid intersection problem when the solution size is a parameter.
In general, \textit{kernelization} for a parameterized problem is a polynomial-time transformation that reduces any given instance to an equivalent instance of the same problem, called a \textit{kernel}, with size and parameter bounded by a function of the parameter in the input.
Kernelization plays a central role in parameterized complexity theory~\cite{fomin2019kernelization}, as kernelization can be used as a preprocessing in fixed-parameter algorithms, and moreover, the existence of a kernel is equivalent to fixed-parameter tractability of a given parameterized problem.

A kernel is \textit{polynomial} if its size is polynomially bounded by the parameter of the input.
There is a line of research in parameterized complexity theory that focuses on classifying fixed-parameter tractable problems on whether they admit polynomial kernels or not. 
To show the existence of a polynomial kernel or to reduce the kernel size, a variety of sophisticated techniques have been developed~(see e.g., \cite{kratsch2014recent, lokshtanov2012kernelization, fomin2014kernelization, guo2007invitation}).
Also, various complexity-theoretic lower bound tools for kernelization have been established~\cite{bodlaender2009problems, dell2012kernelization, dell2014satisfiability, drucker2015new, fortnow2011infeasibility, hermelin2012weak, kenyon2007rank}.
See also the comprehensive textbook by Fomin et al.~\cite{fomin2019kernelization}.

Recently, Huang and Ward~\cite{huang2023fpt} showed that the weighted $d$-matroid intersection problem admits a kernel of size $O(d^{(k - 1)d})$, 
where the parameter $k$ is the solution size.
It is, however, open whether it has a polynomial kernel or not when $d$ is a fixed constant.
There are only a few special cases admitting polynomial-size kernels in the literature as below.

One of specific special cases is the $d$-dimensional matching problem, which is the $d$-matroid intersection problem in which all matroids are simple partition matroids.
The problem is a special case of the $d$-matching problem~(also known as the $d$-set packing problem), and the $d$-matching problem is known to admit a kernel of size $O(k^d)$.
Thus, the $d$-dimensional matching problem has a kernel of size $O(k^d)$.
The proof makes use of the sunflower lemma~(see~\cite[Chapter 8]{fomin2019kernelization}).
Dell and Marx~\cite{dell2012kernelization} showed that the kernel size $O(k^d)$ is asymptotically tight even for the $d$-dimensional matching problem.

Another special case is when given matroids are linearly representable.
It is known that this case has a randomized FPT algorithm parameterized by $d$ and the solution size $k$~(see~\cite[Chapter 12.3.4]{cygan2015parameterized}), which yields a kernel consisting of $O(k^d)$ elements of the ground set.
The algorithm is based on the representative sets technique (see~\cite[Chapter 11]{fomin2019kernelization}), specifically designed for linear matroids.
In recent years, representative sets have become one of the most widely-studied tools in the parameterized algorithms community, which can effectively deal with linear matroids~\cite{marx2009parameterized,fomin2016efficient,kratsch2020representative, wahlstrom2024representative,fomin2014representative}.

In the above mentioned spacial cases, a polynomial kernel is obtained via well-established kernelization techniques such as the sunflower lemma and representative sets.
However, the existing techniques rely on how given matroids are represented.
In fact, they fail to handle the case when even one matroid is given by an independence oracle.

\subsection{Our Contribution} \label{subsec:ourcontribution}

Our main contribution is to develop a new kernelization technique for handling general matroids.
Using this technique, we design kernelization algorithms for a broader class of the weighted $d$-matroid intersection problem, 
where one matroid is arbitrary and given by an independence oracle, while the remaining $d-1$ matroids are ``simple'' matroids such as 
partition, graphic, cographic, transversal, or laminar matroids.
The matroid classes discussed in this paper are summarized in Appendix~\ref{sec:matroidclass}.
Throughout the paper, $d$ is regarded as a fixed constant. 

The kernel that our algorithm finds satisfies an exchange property for matroids, which we call a \textit{reachable kernel}.
Intuitively, a reachable kernel satisfies the condition that any feasible solution of a given instance can reach a better solution in the kernel by exchanging elements one by one with constant probability.

Below is the formal definition of a reachable kernel.
\begin{definition}\label{def:weakcoreset}
Let $E$ be a finite set and $\mathcal{F} \subseteq 2^E$ be 
a family of feasible sets. 
Let $w\colon E \to \mathbb{R}_{\ge 0}$ be a non-negative weight function
and let $k$ be a positive integer. 
We say that a randomly selected set $R \subseteq E$ is a {\em \weakcoreset with respect to $\mathcal{F}$, $w$, and $k$}
if the following condition holds:
\begin{quote}
for any $\{x_1, \dots , x_t\} \in \mathcal{F}$ with $t \le k$, 
there exist elements $y_1, \dots ,y_t \in R$ such that 
$w(y_i) \ge w(x_i)$ and $\{y_1, \dots , y_{i}, x_{i+1}, \dots , x_t\}$ is a set of size $t$ in $\mathcal{F}$
for any $i \in \{1,2,\dots,t\}$ with probability at least $\frac{2}{3}$.\footnote{The value $\frac{1}{3}$ of the error probability is 
not essential. Indeed, by computing a \weakcoreset $c$ times and taking their union, 
we can reduce the error probability to less than $(1/3)^c$, increasing the output size by a factor of $c$.}
\end{quote}
When $\mathcal{F}$, $w$, and $k$ are clear, we simply call it a {\em \weakcoreset}. 
\end{definition}
For the problem of finding a set $X \subseteq E$ of maximum weight subject to $X \in \F$, 
by the definition, an instance restricted to the \weakcoreset $R$ has the same optimal value as the given instance when the solution size is at most $k$.
In this sense, the set $R$ can be regarded as a~(randomized) kernel under the assumption that the membership oracle for $\F$ is available and it does not contribute to the input size\footnote{The assumption is standard in the literature of kernelization for matroid-constrained problems~\cite{banik_et_al:LIPIcs.ESA.2024.17}.}.

As we will discuss later in Section~\ref{para:anotherapplication}, 
reachable kernels not only guarantee a reduced input instance for the purposes of FPT, 
but they also find additional applications in other computational models. 
In the following, we discuss the size of a \weakcoreset that can be guaranteed in various contexts.  

We first consider the case when we are given one arbitrary matroid and $d-1$ simple partition matroids.
Here, a matroid $\mathcal{M} = (E, \mathcal{I})$ is called a {\em simple partition matroid} if 
there exists a partition $(E^1, \ldots, E^{\ell})$ of $E$ such that $I \in \I$ 
if and only if $|I \cap E^j| \le 1$ for any $j \in [\ell]$.

\begin{theorem}\label{thm:partition}
For the weighted $d$-matroid intersection problem such that, out of the $d$ given matroids, $d-1$ are simple partition matroids, we can construct a \weakcoreset of size\footnote{The notation $\tilde{O}$ hides a polylogarithmic function of $k$.} $\tilde{O}(k^d)$ in polynomial time.
\end{theorem}

Theorem~\ref{thm:partition} implies that the $d$-dimensional matching problem admits a kernel of size $\tilde{O}(k^d)$, as the $d$-dimensional matching problem is identical with the weighted $d$-matroid intersection problem when all the given matroids are simple partition matroids.
As mentioned before, the optimal kernel size of the $d$-dimensional matching problem is $O(k^d)$.
Thus, our result matches the optimal bound, ignoring a logarithmic factor.

Theorem~\ref{thm:partition} can be adapted to more generalized matroids.
To state our result, we introduce a notion of \gkgood matroids, which is of independent interest.
Here, $\mathrm{span}_{\M}(X)$ of a matroid $\M$ for a subset $X\subseteq E$ denotes the set of elements $e \in E$ such that adding $e$ to $X$ does not increase its rank.

\begin{definition}
 For a non-decreasing function $g:\mathbb{Z}_{\ge 0}\to \mathbb{Z}_{\ge 0}$,  
 a matroid $\M = (E, \I)$ is called {\em \gkgood} if it satisfies the following condition:
 for any $X \in \I$ with $|X| \leq k$, there exists $F \subseteq E$ with $|F| \leq g(k)$ such that $\mathrm{span}_{\M}(X) = \bigcup_{f \in F} \mathrm{span}_{\M}(f)$.
\end{definition}

Roughly speaking, this definition means that any independent set $X$ of size at most $k$ can essentially span at most $g(k)$ elements, when we identify ``parallel'' elements with one element.
For example, a simple partition matroid is $k$-coverable, while a uniform matroid of rank $k$ on a large ground set is not \gkgood for any function $g$.

\begin{theorem} \label{thm:gkgood}
For the weighted $d$-matroid intersection problem with $d$ matroids $\M_0, \M_1, \dots, \M_{d-1}$ where $\M_i$ is \gikgood for each $i \in \{1,\dots, d-1\}$, 
we can construct a \weakcoreset of size $\tilde{O} \left( k \cdot \prod_i g_i(k) \right)$ in polynomial time.
\end{theorem}

Since a simple partition matroid is $k$-coverable, Theorem~\ref{thm:gkgood} is a generalization of Theorem~\ref{thm:partition}. 
Furthermore, since graphic matroids are \gkgoodf{$O(k^2)$} and cographic matroids are \gkgoodf{$O(k)$} as we will see in Section~\ref{sec:analysisgk}, 
Theorem~\ref{thm:gkgood} implies the following corollary.

\begin{corollary} \label{cor:graphic}
For the weighted $d$-matroid intersection problem such that, out of the $d$ given matroids, all except one are graphic (resp.~cographic), 
we can construct a \weakcoreset of size $\tilde{O}(k^{2d-1})$ (resp.~$\tilde{O}(k^{d})$) in polynomial time.
\end{corollary}

Although graphic and cographic matroids are $g(k)$-coverable, it may be somehow unexpected that relatively mild generalizations of simple partition matroids, namely, 
transversal and laminar matroids are not (since they both generalize uniform matroids). To obtain the following results, significantly different 
techniques are required.

\begin{restatable}{theorem}{transversal}\label{thm:transversal}
For the weighted $d$-matroid intersection problem such that, out of the $d$ given matroids, all except one are transversal,
we can construct a \weakcoreset of size $\tilde{O}(k^{d})$ in polynomial time.    
\end{restatable}

\begin{restatable}{theorem}{laminar}\label{thm:laminar}
For the weighted $d$-matroid intersection problem such that, out of the $d$ given matroids, all except one are laminar matroids,
we can construct a \weakcoreset of size $k^{O(d \log k)}$ in polynomial time.
\end{restatable}

Theorems~\ref{thm:transversal} and~\ref{thm:laminar}  are proved in Sections~\ref{sec:transversal} and~\ref{sec:laminar}, respectively. For transversal matroids, we introduce a reduction technique to transform them to simple partition matroids. On the other hand, for laminar matroids, as there are a significantly larger number of possibilities concerning how a feasible solution is ``distributed'' among the sets of the laminar family, we need to employ a more aggressive guessing strategy and this causes our kernel size to be quasi-polynomial.

Our kernelization technique can also be applied to certain classes of the {\em $d$-matchoid problem}. 
Since dealing with the general $d$-matchoid problem would make the discussion overly complicated, 
this paper focuses only on one particularly interesting example, 
the {\em weighted matching problem with a matroid constraint}.
In the problem, we are given a graph $G = (V, E)$, a matroid $\M = (E, \I)$ defined on the edge set of $G$, and a weight function $w: E\to \mathbb{R}_{\ge 0}$. 
The goal of the problem is to find a matching of maximum weight that is independent in $\M$. 
For this problem, we present a kernerization algorithm parameterized by the solution size $k$. 
See Section~\ref{sec:matroidal_matching} for the proof. 

\begin{restatable}{theorem}{matroidalmatching}\label{thm:matroidalmatching}
    For the weighted matching problem with a matroid constraint, we can construct a \weakcoreset of size $\tilde{O}(k^3)$ in polynomial time.
\end{restatable}

The weighted matching problem with a matroid constraint is a generalization of the {\em rainbow matching problem}~\cite{pfender2014complexity, gupta2019parameterized,gupta2020quadratic}, 
which corresponds to the case where the given matroid is a simple partition matroid.
The rainbow matching problem has a kernel of size $O(k^3)$ by~\cite[Theorem 3]{gupta2019parameterized}, which is later improved to
a kernel with $O(k^2)$ vertices and $O(k^3)$ edges~\cite{gupta2020quadratic}.

Finally, we discuss the case when the reachability is realized deterministically.
We say that a subset $R \subseteq E$ is a {\rm \coreset} if 
it satisfies the conditions in Definition~\ref{def:weakcoreset} with probability $1$.
See Section~\ref{sec:deterministic} for details.

\begin{restatable}{theorem}{deterministic}\label{thm:deterministic}
For the weighted $2$-matroid intersection problem such that the matroids are a \gkgood matroid and an arbitrary matroid,
we can construct a \coreset of size $k^2 (g(k)+1)$ in polynomial time.
\end{restatable}

\subsection{Advantages of Reachable Kernel} \label{para:anotherapplication}

We here describe a few applications of a \weakcoreset in different contexts.

The first application is in the design of {\em parameterized streaming algorithms}~\cite{fafianie2014streaming, chitnis2014parameterized, chitnis2016kernelization, chitnis_et_al:LIPIcs.IPEC.2019.7}.
Suppose that we are given a feasible region $\mathcal{F} \subseteq 2^E$ over a ground set $E$ through oracle access, and the elements of $E$ arrive one by one. 
The goal is to find a set $X \subseteq E$ that maximizes $\sum_{e \in X} w(e)$ 
subject to $X \in \mathcal{F}$ and $|X| \le k$. 
We consider the setting where $|E|$ is very large and at any moment we can keep at most $f(k)$ elements, discarding all the others. 
Huang and Ward~\cite{huang2023fpt} introduced a slight generalization of the concept of a \weakcoreset, which they call a {\em representative set}\footnote{Note that this notion is different from representative sets used for the linear matroid intersection.}.
For the $d$-matroid intersection problem (and, more generally, the $d$-matchoid problem), they proposed a streaming algorithm that always maintains a representative set whose size is bounded by an exponential function of $k$.
In this context, we can confirm that our results can be generalized to adapt the streaming setting, resulting in streaming algorithms that use only polynomial~(or quasi-polynomial) space in $k$ for several classes of matroids.

Furthermore, our \weakcoreset has applications to improving the running time of EPTASs for the budgeted version of matroid intersection~\cite{doronarad_et_al:LIPIcs.ICALP.2023.49}.
Doron-Arad et al.~\cite{doronarad_et_al:LIPIcs.ICALP.2023.49} presented an EPTAS for the budgeted version, using a representative set by Huang and Ward~\cite{huang2023fpt}.
Their EPTAS runs in $q(\varepsilon)^{O(q(\varepsilon))} \cdot \mathrm{poly}(|I|)$ time, where $q(\eps) = \left\lceil \frac{1}{\eps^{\eps^{-1}}} \right\rceil$ and $|I|$ denotes the input size.
Although this running time yields an EPTAS, the dependence on $\varepsilon$ is extremely high.
By following the argument in~\cite{doronarad_et_al:LIPIcs.ICALP.2023.49} with our reachable kernel, replaced with their representative set, we can improve the running time of the EPTAS to $2^{O\left( \eps^{-2} \log \frac{1}{\eps}\right)} \cdot \text{poly}(|I|)$ for several classes of matroids, where the dependence on $\varepsilon$ becomes significantly milder.

\subsection{Additional Related Work}

Matroids play an important role in parameterized complexity theory, as the textbooks~\cite{cygan2015parameterized, fomin2019kernelization} devote chapters to matroids.
See also~\cite{lokshtanov2018deterministic, gurjar_et_al:LIPIcs.MFCS.2025.54, kratsch2014compression, bentert_et_al:LIPIcs.ESA.2025.83, marx2009parameterized,fomin2016efficient,kratsch2020representative, wahlstrom2024representative,fomin2014representative, fomin_et_al:LIPIcs.SWAT.2024.22, banik_et_al:LIPIcs.ESA.2024.17, lokshtanov_et_al:LIPIcs.ITCS.2018.32, fomin_et_al:LIPIcs.ICALP.2019.59}.
In particular, linear matroids have been actively studied in the parameterized algorithm community.
For the $d$-matroid intersection problem when all matroids are linear and matrix representations are available,
Barvinok~\cite{barvinok1995new} presented the first algorithm faster than naive enumeration.
Marx~\cite{marx2009parameterized} provided a fixed-parameter algorithm, and subsequent works~\cite{fomin2016efficient, brand2021parameterized, eiben2024determinantal} proposed faster algorithms.

Recently, there have been several studies addressing graph problems under matroid constraints, with a focus on parameterized algorithm and kernelization, such as the stable set problem~\cite{fomin_et_al:LIPIcs.SWAT.2024.22} and various cut problems~\cite{banik_et_al:LIPIcs.ESA.2024.17}.

%% file: preliminaries.tex
\section{Preliminaries} \label{sec:preliminaries}

\subsection{Basic Notation}

For a positive integer $t$, we denote $[t]=\{1, 2, \dots , t\}$. 
Throughout this paper, we often write $A+x := A\cup\{x\}$ and $A-x := A\setminus \{x\}$ for a set $A$ and an element $x$.

A pair $\mathcal{M}=(E, \I)$ of a finite set $E$ and a non-empty set family $\I \subseteq 2^E$ is called a \textit{matroid} if the following properties are satisfied.
\begin{enumerate}
\item If $X\in \I$ and $Y \subseteq X$ then $Y\in \I$.
\item If $X, Y \in \I$ and $|Y| < |X|$, then there exists $x \in X\setminus Y$ such that $Y +x \in \I$.
\end{enumerate}
A set $X\in \I$ is called an \textit{independent set}.
We call a maximal independent set a \textit{base}.
For $U \subseteq E$, the {\em restriction} of $\M$ to $U$ is the matroid $\M | U = (U, \I_U)$, where 
$\I_U = \{ X \in \I \mid X \subseteq U \}$. 
Throughout the paper, we assume that every matroid has no loop, i.e., $\{x\} \in \I$ for every $x \in X$, 
since we can remove elements that do not belong to any independent set.

Let $\M=(E, \I)$ be a matroid.
For a subset $X$ of $E$, we define the \textit{rank} of $X$ as $\textrm{rank}_{\M} (X) = \max\{|Y| \mid Y\subseteq X, Y\in \I\}$.
The \textit{rank} of $\M$ is equal to $\textrm{rank}_\M (E)$.
For a subset $X$ of $E$, the \textit{span} of $X$ is defined as
\[
\textrm{span}_{\M} (X) = \{ x\in E\mid \textrm{rank} (X)=\textrm{rank} (X + x)\}.
\]
When a given matroid $\M$ is clear from the context, we sometimes write $\textrm{rank} (X)$ and $\textrm{span} (X)$ for $\textrm{rank}_\M (X)$ and $\textrm{span}_\M (X)$, respectively.
For an element $e \in E$, $\textrm{span}_{\M} (\{e\})$ is simply denoted by $\textrm{span}_{\M} (e)$. 
A pair of distinct elements $e, e' \in E$ are called {\em parallel} if $\textrm{span}_{\M} (e) = \textrm{span}_{\M} (e')$. 

We assume that each general matroid is given by an independence oracle, while each special matroid is given in an appropriate explicit form. 
Classes of matroids considered in this paper are summarized in Appendix~\ref{sec:matroidclass}.

\subsection{Technical Tools}

Let $\F$ be a family of subsets in $E$.
As in Definition~\ref{def:weakcoreset}, a \weakcoreset for $\F$ is defined by the existence of a sequence $y_1, \dots , y_t$ of elements. 
We observe that a randomly selected set $R \subseteq E$ is a \weakcoreset if $R$ satisfies the following single exchangeability.
\begin{description}
   \item[(SingleEXC)]
    For any $X \in \mathcal{F}$ with $|X| \le k$ and any $x \in X$, there exists $y \in R \setminus (X-x)$ such that 
    $w(y) \ge w(x)$ and $X - x + y \in \mathcal{F}$ with probability at least $1 - \frac{1}{3k}$.
\end{description}
Indeed, for $i = 1, 2, \dots , t$ in this order, by applying (SingleEXC) for 
$X = \{y_1, \dots , y_{i-1}, x_i, \dots , x_t\}$ and $x = x_i$, 
we can obtain $y_i$ satisfying the conditions in Definition~\ref{def:weakcoreset}. 
Note that the total error probability is at most $\frac{1}{3k} \cdot t \le \frac{1}{3}$. 
We remark that, when $k = 1$, (SingleEXC) obviously holds by choosing a maximum-weight element of size one in $\F$ as $R$.
Thus, in what follows, we may assume that $k \ge 2$.

If $\mathcal{F}$ is a family of the independent sets of a single matroid, a simple greedy algorithm can find a maximum-weight independent set of size at most $k$, which satisfies~(SingleEXC). 
As we will perform this greedy algorithm to a subset $F\subseteq E$ as a subroutine of our algorithm, we describe it formally in Algorithm~\ref{greedy}.
Here, we assume that the elements are sorted in non-increasing order of their weights, 
and we omit the weight function $w$ in the input of the algorithm.

\begin{algorithm}
    \caption{\texttt{Greedy}$(\M, F, k)$}\label{greedy}
    \KwInput{A matroid $\M = (E, \I)$, $F \subseteq E$, and an integer $k$}
    Let $F = \{ f_1, \ldots, f_\ell \}$; \\
    \tcp{The elements in $F$ are sorted in non-increasing order of the weight, i.e., $w(f_1) \geq w(f_2) \geq \cdots \geq w(f_\ell)$.} 
    $I \gets \emptyset$; \\
    \For{$i \gets 1$ \KwTo $\ell$} {
        \If{$I + f_i \in \I$ and $|I + f_i| \leq k$} {
            $I \gets I + f_i$; \\
        }
    }
    \Return $I$;
\end{algorithm}

The following two lemmas are basic properties of matroids and the greedy algorithm.
We provide proofs for completeness. 

\begin{lemma} \label{lem:matroid_property}
    Let $\M=(E,\I)$ be a matroid. Suppose that $T, X\in\I$ and $x \in X$ with $x \in \mathrm{span}_\M(T)$.
    Then, there exists $y\in T \setminus (X - x)$ such that $X - x + y \in \I$.
\end{lemma}
    
\begin{proof}
    Assume for contradiction that $X - x + y \notin \I$ for every $y \in T \setminus (X - x)$.
    This means that $y \in \mathrm{span}_\M(X-x)$ for all $y \in T$, and hence $T\subseteq \mathrm{span}_\M(X - x)$.
    By taking the span of both sides, we obtain
    \begin{equation*}
        x\in \mathrm{span}_\M(T)\subseteq \mathrm{span}_\M\!\big(\mathrm{span}_\M(X-x)\big) =\mathrm{span}_\M(X-x).
    \end{equation*}
    Hence $x\in \mathrm{span}_\M(X-x)$, which implies that $X$ is not independent, a contradiction.
    Therefore, some $y \in T$ satisfies $X - x + y \in \I$. 
\end{proof}

\begin{lemma} \label{lem:greedyproperty}
    Let $\M=(E,\I)$ be a matroid with a weight function $w \colon E \to \mathbb{R}_{\ge 0}$ and let $k$ be a positive integer.  
    Suppose that $F \subseteq E$, $X\in\I$, $|X| \le k$, and $x \in F \cap X$. 
    Then, the output of {\rm $\texttt{Greedy}(\M, F, k)$} contains an element $y$ such that 
    $y \notin X-x$, $w(y) \ge w(x)$, and $X-x+y \in \mathcal{I}$.     
\end{lemma}

\begin{proof}
    Let $R$ be the output of $\texttt{Greedy}(\M, F, k)$. 
    It suffices to consider the case where $x \not\in R$, since otherwise $y = x$ satisfies the conditions. 
    Since $x \not\in R$ means that the greedy algorithm does not choose $x$, 
    $T := \{ e \in R \mid w(e) \ge w(x) \}$ satisfies that 
    $x \in \mathrm{span}_\M(T)$ or $|T| = k$. 
    If $x \in \mathrm{span}_\M(T)$, then Lemma~\ref{lem:matroid_property} shows the existence of $y \in T \setminus (X - x)$ with $X - x + y \in \I$.
    If $|T| = k$, then we have $|T| > |X-x|$ and $T, X-x \in \I$, and hence there exists $y \in T \setminus (X - x)$ with $X - x + y \in \I$ by the definition of a matroid. 
    Since $y \in T$ implies $y \in R$ and $w(y) \ge w(x)$, a desired $y$ exists in both cases. 
\end{proof}

For a matroid $\mathcal{M} = (E, \mathcal{I})$ and for $X \in \mathcal{I}$ and $x \in X$, 
we say that $Y \subseteq E$ is an {\em $(X, x)$-exchangeable set with respect to $\M$} if 
$x \in Y$ and $X -x +y \in \mathcal{I}$ for every $y \in Y$. 
In matroid-theoretic terms, an $(X, x)$-exchangeable set is a subset of a fundamental cocircuit of the matroid obtained from $\M$ by truncating it to rank $|X|$.   
In the later sections, we repeatedly use the following lemma.

\begin{lemma} \label{lem:exchangeablesetgreedy}
    Suppose we are given a matroid $\M_i=(E,\I_i)$ for $i \in \{0, 1, \dots, d-1\}$, a weight function $w \colon E \to \mathbb{R}_{\ge 0}$, and a positive integer $k$. 
    Let $\F = \bigcap_{i=0}^{d-1}\I_i$, and let $X \in \F$ be a feasible set with $|X| \le k$, and $x \in X$. 
    If $F_i \subseteq E$ is an $(X, x)$-exchangeable set with respect to $\M_i$ for $i \in [d-1]$, then
    the output of {\rm $\texttt{Greedy}(\M_0, \bigcap_{i \in [d - 1]} F_i, k)$} contains an element $y$ 
    such that 
    $y \notin X - x$, $w(y) \ge w(x)$, and $X-x+y \in \F$.         
\end{lemma}

\begin{proof}
Let $R$ be the output of $\texttt{Greedy}(\M_0, \bigcap_{i \in [d - 1]} F_i, k)$. 
Since $x \in \bigcap_{i \in [d - 1]} F_i$, by Lemma~\ref{lem:greedyproperty}, there exists $y \in R \setminus (X-x)$ such that $w(y) \geq w(x)$ and $X - x + y \in \I_0$.
Such an element $y$ satisfies that $X - x + y \in \I_i$ for each $i \in [d-1]$, because $y \in F_i$ and $F_i$ is an $(X, x)$-exchangeable set with respect to $\M_i$. 
Therefore, $X - x + y \in \bigcap_{i=0}^{d-1}\I_i = \F$, which completes the proof. 
\end{proof}

%% file: partition.tex
\section{Partition Matroids and Extensions} \label{sec:partition}

\subsection{Kernelization Algorithm for Simple Partition Matroids} \label{subsec:partition}

In this subsection, we prove Theorem~\ref{thm:partition}.
Recall that a matroid $\mathcal{M} = (E, \mathcal{I})$ is a {\em simple partition matroid} if 
there exists a partition $(E^1, \ldots, E^{\ell})$ of $E$  such that
$\I = \{ I \subseteq E \mid |I \cap E^j| \leq 1 \text{ for any } j \in [\ell] \}$.
Suppose that we are given $d$ matroids $\M_0,\M_1,\dots,\M_{d-1}$, where $\M_0 = (E, \I_0)$ is an arbitrary matroid and $\M_i = (E, \I_i)$ is a simple partition matroid for $i \in [d-1]$.
We assume that each simple partition matroid $\M_i$ is represented by a partition $(E_i^1, \ldots, E_i^{\ell_i})$
of $E$. 
Define $\F=\bigcap_{i=0}^{d-1}\I_i$, which is a family of feasible sets of the problem.
Our goal is to construct a \weakcoreset of size $\tilde{O}(k^d)$. 

We construct a \weakcoreset by repeating the following procedure $T = \tilde{O}(k^{d - 1})$ times, where the exact value of $T$ will be determined later.
For each $i \in [d - 1]$, we sample each element $j \in [\ell_i]$ independently with probability $1/k$ to obtain a random subset $S_i \subseteq [\ell_i]$, and  
let $F_i = \bigcup_{j \in S_i} E^j_i$.
Then we apply $\texttt{Greedy}(\M_0, \bigcap_{i \in [d - 1]} F_i, k)$, that is, 
we greedily select up to $k$ elements from $\bigcap_{i \in [d - 1]} F_i$ in non-increasing order of weight, so that the selected set is independent in $\M_0$.
We denote the resulting set by $R'$ and add it to our \weakcoreset $R$.
The pseudocode for this algorithm is shown in Algorithm~\ref{alg:parititionkernel}.

\begin{algorithm}
    \caption{\texttt{PartitionKernel}$(\mathcal{M}_0, \dots , \mathcal{M}_{d-1})$}\label{alg:parititionkernel}
    $R \gets \emptyset$; \\
    \RepTimes{$T$}{
        \For{$i \in [d - 1]$} {
            Let $S_i \subseteq [\ell_i]$ be a subset obtained by independently sampling each element $j \in [\ell_i]$ with probability $\frac{1}{k}$; \\
            $F_i \gets \bigcup_{j \in S_i} E^j_i$; \\ 
        }
        $R' \gets \texttt{Greedy}(\M_0, \bigcap_{i \in [d - 1]} F_i, k)$; \\
        $R \gets R \cup R'$; \\
    }
    \Return $R$;
\end{algorithm}

In what follows, we prove the correctness of the algorithm by showing that the obtained set $R$ satisfies~(SingleEXC).
Let $X \in \mathcal{F}$ be a feasible solution of size at most $k$, and $x$ be an arbitrary element in $X$.

Consider one round (lines 3--7) in the algorithm where we generate the sets $F_i$'s for $i\in [d-1]$ by sampling.
We say that the round is {\em successful} if 
$F_i$ is an $(X, x)$-exchangeable set with respect to $\M_i$ (i.e., $x \in F_i$ and $X-x+y \in \I_i$ for any $y \in F_i$) for every $i \in [d-1]$. 
We now evaluate the probability that one round is successful. 

For each $i \in [d - 1]$, let $j^*_i \in [\ell_i]$ be the index with $x \in E^{j^*_i}_i$. 
Then, by the definition of a simple partition matroid, 
$F_i$ is an $(X, x)$-exchangeable set with respect to $\M_i$  
if and only if $j^*_i \in S_i$ and $j \notin S_i$ for any $j \in [\ell_i]$ with $(X-x) \cap E^j_i \neq \emptyset$. 
Hence, the probability that $F_i$ is an $(X, x)$-exchangeable set with respect to $\M_i$ is at least 
\begin{equation*}
\frac{1}{k} \left( 1 -  \frac{1}{k} \right)^{|X - x|} \geq \frac{1}{k} \left( 1 -  \frac{1}{k} \right)^{k - 1} \geq \frac{1}{ek},   
\end{equation*}
where we use the fact that $(1 - 1/\alpha)^{\alpha - 1} \geq 1/e$ for $\alpha > 1$ and the assumption that $k \ge 2$.
Therefore, the probability that one round is successful is at least $\left( \frac{1}{ek} \right)^{d - 1}$.
By setting $T = \lceil (ek)^{d - 1} \log(3k) \rceil$, 
the probability that there exists no successful round among the $T$ rounds is at most 
\begin{equation*}
\left( 1 - \left( \frac{1}{ek} \right)^{d - 1} \right)^T 
\leq \frac{1}{3k}, 
\end{equation*}
where we use $(1 - 1/\alpha)^{\alpha} \leq 1/e$ for $\alpha \ge 1$.
Therefore, with probability at least $1 - \frac{1}{3k}$, there exists at least one successful round.

For the sets $F_i$'s generated in a successful round, Lemma~\ref{lem:exchangeablesetgreedy} shows that
the output $R'$ of $\texttt{Greedy}(\M_0, \bigcap_{i \in [d - 1]} F_i, k)$ contains an element $y$ such that 
$y \notin X-x$, $w(y) \geq w(x)$, and $X - x + y \in \F$.
Since $R \subseteq E$ returned by Algorithm~\ref{alg:parititionkernel} includes $R'$, such an element $y$  
satisfies the conditions in (SingleEXC). 

Since this argument works for any $X \in \F$ and any $x \in X$, 
$R$ satisfies (SingleEXC).
The size of the set $R$ is at most $k \cdot T = O(e^{d - 1}k^d \log k)$.
Since Algorithm~\ref{alg:parititionkernel} clearly runs in polynomial time, the proof of Theorem~\ref{thm:partition} is complete.

\subsection{Kernelization Algorithm for \gkgood Matroids}  \label{subsec:algoforgkgood}

In this subsection, we consider \gkgood matroids, and prove Theorem~\ref{thm:gkgood}. 
Let $\M_0= (E, \I_0)$ be an arbitrary matroid, and $\M_i= (E, \I_i)$ be a \gikgood matroid for $i \in [d-1]$.
Note that $g_i(k) \ge k$ by the definition of a \gikgood matroid. 
We may assume that $|E| = \Omega \left( k \cdot \prod_i g_i(k) \right)$, since otherwise $E$ is trivially a \weakcoreset of size $\tilde{O} \left( k \cdot \prod_i g_i(k) \right)$.

We observe that, for $i\in [d-1]$, the ground set $E$ is covered by the spans of singletons, that is, there exists a subset $\tilde{E}_i\subseteq E$ such that $E=\bigcup_{e\in \tilde{E}_i}{\rm span}_{\M_i} (e)$ and ${\rm span}_{\M_i} (e)\neq {\rm span}_{\M_i} (e')$ for any distinct $e, e'\in \tilde{E}_i$.
Here, $\tilde{E}_i$ consists of one representative from each parallel class of $\M_i$.
We note that $\{{\rm span}_{\M_i} (e)\mid e\in \tilde{E}_i\}$ forms a partition of $E$ such that every independent set $X$ contains at most one element from each part of the partition.
We also note that the converse does not hold; even if a set $X$ contains at most one element from each part of the partition, it is not necessarily independent. 

Our algorithm that computes a \weakcoreset for \gkgood matroids is based on the same ideas as that for simple partition matroids in Section~\ref{subsec:partition}.
We will use the partition $\{{\rm span}_{\M_i} (e)\mid e\in \tilde{E}_i\}$ to sample a random subset.

\begin{algorithm}
    \caption{\texttt{CoverableKernel}$(\mathcal{M}_0, \dots , \mathcal{M}_{d-1})$}\label{alg:gkgoodkernel}
    $R \gets \emptyset$; \\
    \RepTimes{$T$}{
        \For{$i \in [d - 1]$} {
            Let $S_i \subseteq \tilde{E}_i$ be a subset obtained by independently sampling each element $e \in \tilde{E}_i$ with probability $\frac{1}{g_i(k)}$;  \\
            $F_i \gets \bigcup_{e \in S_i} \span_{\M_i}(e)$; \\ 
        }
        $R' \gets \texttt{Greedy}(\M_0, \bigcap_{i \in [d - 1]} F_i, k)$; \\
        $R \gets R \cup R'$; \\
    }
    \Return $R$;
\end{algorithm}

In our algorithm, we repeat the following procedure $T = \tilde{O}(\prod_{i = 1}^{d - 1} g_i(k))$ times. 
For each $i \in [d - 1]$, we sample each element $e \in \tilde{E}_i$ independently with probability $1/g_i(k)$ to obtain a random subset $S_i \subseteq \tilde{E}_i$.
Let $F_i = \bigcup_{e \in S_i} \span_{\M_i}(e)$.
Then, we apply $\texttt{Greedy}(\M_0, \bigcap_{i \in [d - 1]} F_i, k)$. 
We denote the resulting set by $R'$ and add it to our \weakcoreset $R$.
The pseudocode for this algorithm is shown in Algorithm~\ref{alg:gkgoodkernel}.

Similarly to the proof of Theorem~\ref{thm:partition}, we prove that $R$ satisfies~(SingleEXC).  
Let $X \in \mathcal{F}$ be a feasible solution of size at most $k$, and $x$ be an arbitrary element in $X$.
Consider one round in the algorithm where we generate the sets $F_i$'s for $i\in [d-1]$ by sampling.
We say that the round is {\em successful} if 
$F_i$ is an $(X, x)$-exchangeable set with respect to $\M_i$ (i.e., $x \in F_i$ and $X-x+y \in \I_i$ for any $y \in F_i$) for every $i \in [d-1]$. 
We now evaluate the probability that one round is successful. 

For each $i \in [d - 1]$, let $f_i \in \tilde{E}_i$ be the element with $x \in \span_{\M_i}(f_i)$ and 
let $Y_i \subseteq \tilde{E}_i$ be the set with $\span_{\M_i}(X-x) = \bigcup_{f \in Y_i} \span_{\M_i}(f)$.
Note that $f_i$ and $Y_i$ are uniquely determined and $f_i \notin Y_i$. 
Note also that $|Y_i| \le g_i(k)$ by the definition of a \gikgood matroid. 
Observe that 
$F_i = \bigcup_{e \in S_i} \span_{\M_i}(e)$ is an $(X, x)$-exchangeable set with respect to $\M_i$  
if and only if $f_i \in S_i$ and $e \notin S_i$ for any $e \in Y_i$. 
Hence, the probability that $F_i$ is an $(X, x)$-exchangeable set with respect to $\M_i$ is at least 
\begin{equation*}
\frac{1}{g_i(k)} \left( 1 -  \frac{1}{g_i(k)} \right)^{|Y_i|} \geq \frac{1}{g_i(k)} \left( 1 -  \frac{1}{g_i(k)} \right)^{g_i(k)} \geq 
\frac{1}{4 \cdot g_i(k)},
\end{equation*}
where the second inequality follows by the fact that $(1 - 1/\alpha)^{\alpha}$ is a monotonically increasing function for $\alpha > 1$ 
and by the assumption that $g_i(k) \ge k \ge 2$.
Therefore, the probability that one round is successful is at least $\frac{1}{4^{d - 1}}  \prod_{i = 1}^{d - 1} \frac{1}{g_i(k)}$.
By setting $T = \lceil 4^{d - 1} \log(3k) \prod_{i = 1}^{d - 1} g_i(k) \rceil$,  
the probability that there exists no successful round among the $T$ rounds is at most 
\begin{equation*}
\begin{aligned}
&\left( 1 -  \frac{1}{4^{d - 1}} \prod_{i = 1}^{d - 1} \frac{1}{g_i(k)} \right)^T \leq \frac{1}{3k}, 
\end{aligned}
\end{equation*}
where we use $(1 - 1/\alpha)^{\alpha} \leq 1/e$ for $\alpha \ge 1$. 
Therefore, with probability at least $1 - \frac{1}{3k}$, there exists at least one successful round.

In the same way as the proof of Theorem~\ref{thm:partition}, 
by applying Lemma~\ref{lem:exchangeablesetgreedy} for the sets $F_i$'s generated in a successful round, 
there exists $y \in R \setminus (X-x)$ such that $w(y) \geq w(x)$ and $X - x + y \in \F$.
Therefore, $R$ satisfies (SingleEXC).

Moreover, Algorithm~\ref{alg:gkgoodkernel} runs in polynomial time in $|E|$, because we assume that $|E| = \Omega \left( k \cdot \prod_i g_i(k) \right)$.
Finally, the size of the set $R$ returned by Algorithm~\ref{alg:gkgoodkernel} is at most 
$k \cdot T = O\left(4^{d - 1}k \log(k) \prod_{i = 1}^{d - 1} g_i(k) \right)$.
This completes the proof of Theorem~\ref{thm:gkgood}.

\subsection{Examples of \gkgood Matroids}
\label{sec:analysisgk}

We have already mentioned that simple partition matroids are \gkgoodf{$k$}, 
which implies that Theorem~\ref{thm:gkgood} is an extension of Theorem~\ref{thm:partition}. 
In this subsection, we present examples of \gkgood matroids. 

A matroid $\mathcal{M} = (E, \mathcal{I})$ is called a {\em graphic matroid} if 
there exists a graph $G = (V, E)$ such that $F \subseteq E$ is an independent set of $\M$ if and only if $F$ forms a forest in $G$. 

\begin{lemma} \label{lem:graphicgk}
    Graphic matroids are \gkgoodf{$O(k^2)$}.
\end{lemma}

\begin{proof}
    Let $\M = (E, \I)$ be a graphic matroid represented by a graph $G=(V, E)$. 
    Fix $X \in \I$ with $|X| \leq k$.
    In a graphic matroid, an edge $e=\{u,v\}\in E$ lies in $\mathrm{span}_{\mathcal{M}}(X)$ if and only if $u$ and $v$ are in the same connected component of the forest $(V,X)$.
    Thus, by removing parallel edges from $\mathrm{span}_{\M}(X)$, we obtain a set $F \subseteq \mathrm{span}_{\M}(X)$ such that $\mathrm{span}_{\M}(X) = \bigcup_{f \in F} \mathrm{span}_{\M}(f)$ and $|F| \leq \binom{k}{2}$.
    This shows that graphic matroids are \gkgoodf{$O(k^2)$}. 
\end{proof}

The dual of a graphic matroid is called a {\em cographic matroid}. 
In other words, $\M = (E, \I)$ is a cographic matroid
if there exists a graph $G=(V, E)$ such that $F \subseteq E$ belongs to $\I$ 
if and only if $G - F$ has the same number of connected components as $G$.

\begin{lemma} \label{lem:cographicgk}
    Cographic matroids are \gkgoodf{$O(k)$}.
\end{lemma}

\begin{proof}
Let $\M = (E, \I)$ be a cographic matroid represented by a graph $G=(V, E)$. 
If $G$ is not connected, we pick two vertices from different components of $G$ and identify them, 
which results in another graph that also represents $\M$. 
By repeating this procedure, we may assume without loss of generality that $G$ is connected. 

If there exists an element $e \in E$ with $\{e\} \notin \I$, then we remove $e$ from $\M$, 
since such an element does not affect the $O(k)$-coverability. 
Note that such an element $e$ is a bridge (i.e., an edge cut of size $1$) in $G$, and 
removing an element $e$ from $\M$ corresponds to contracting an edge $e$ in $G$. 
By repeating this procedure, we may assume without loss of generality that $G$ is $2$-edge-connected. 

Fix $X \in \I$ with $|X| \leq k$. Then, $X \in \I$ means that 
$X$ is an edge set of $G$ such that $G - X$ is connected. 
We shrink each $2$-edge-connected component in $G-X$ into a single vertex, 
which results in a tree $T=(V', E')$. 
Since $e \in E \setminus X$ is in $\mathrm{span}_{\M}(X)$ if and only if 
$G - X - \{e\}$ is disconnected, it holds that $E' = \mathrm{span}_{\M}(X) \setminus X$. 

Let $\varphi\colon V \to V'$ be the mapping that represents the correspondence between $V$ and $V'$, 
that is, the $2$-edge-connected component containing a vertex $v \in V$ in $G$ is shrunk into a vertex $\varphi(v) \in V'$ in $G'$. 
Define $U \subseteq V'$ as 
\[
U = \{v' \in V' \mid \mbox{$X$ contains an edge incident to $\varphi^{-1}(v')$} \},  
\]
where $\varphi^{-1}$ denotes the inverse image of $\varphi$. 
Then, $|U| \le 2|X| \le 2k$. 
Note that every leaf of $T$ is contained in $U$, because $G$ is $2$-edge-connected.  
Let $W \subseteq V'$ be the set of vertices in $T$ whose degree is at least $3$. 
Since $|W|$ is at most the number of leaves in $T$, it holds that $|W| \le |U| \le 2k$. 

Let ${\cal P}$ be the collection of inclusionwise maximal paths in $T$ whose internal vertices are not in $U \cup W$. 
Since $U \cup W$ contains all the leaves and all vertices of degree at least three in $T$, 
the family $\{E(P) \mid P \in {\cal P}\}$ forms a partition of $E'$, where $E(P)$ denotes the set of edges of $P$. 
Note that $|{\cal P}| = |U| + |W| - 1 = O(k)$. 

For any $P \in {\cal P}$ and for any pair of distinct edges $e, e' \in E(P)$, 
$T - \{e, e'\}$ contains a component whose vertex set $C$ consists only of internal vertices of $P$. 
Then, $\varphi^{-1}(C)$ and $V \setminus \varphi^{-1}(C)$ are disconnected in $G - X - \{e, e'\}$.
Furthermore, since $C$ contains no vertex in $U$, no edge in $X$ is incident to $\varphi^{-1}(C)$, 
which implies that $\varphi^{-1}(C)$ and $V \setminus \varphi^{-1}(C)$ are also disconnected in $G - \{e, e'\}$. 
Therefore, $\{e, e'\} \notin \I$, that is, $e$ and $e'$ are parallel in $\M$. 
For each $P \in {\cal P}$, we pick one element $e_{P}$ from $E(P)$, arbitrarily. 
Then, the above argument shows that 
$E(P) \subseteq \mathrm{span}_{\M}(e_P)$ for each $P \in {\cal P}$. 

Let $F = \{e_{P} \mid P \in {\cal P}\}$. 
Then, $|F| = |{\cal P}| = O(k)$, and 
\[
\mathrm{span}_{\M}(X) = X \cup E' 
= X \cup \bigcup_{P \in {\cal P}} E(P)
= \bigcup_{f \in X \cup F} \mathrm{span}_{\M}(f). 
\]
This shows that cographic matroid is \gkgoodf{$O(k)$}. 
\end{proof}

These lemmas show that
Corollary~\ref{cor:graphic} is derived from Theorem~\ref{thm:gkgood}.

%% file: transversal.tex
\section{Kernelization Algorithm for Transversal Matroids}
\label{sec:transversal}

Let $G = (U, W; E)$ be a bipartite graph with a bipartition $(U, W)$. 
The {\em transversal matroid} represented by $G$ is a matroid $\M = (U, \I)$ over $U$ such that 
$X \subseteq U$ is in $\I$ if and only if $G$ has a matching that covers $X$. 
In this subsection, we extend the result for simple partition matroids (Theorem~\ref{thm:partition}) to transversal matroids and prove Theorem~\ref{thm:transversal}.

\transversal*

\begin{proof}
Let $\F$ be the family of subsets in $U$ such that $\F=\bigcap_{i=0}^{d-1}\I_i$, 
where $\M_0 = (U, \I_0)$ is an arbitrary matroid and $\M_i = (U, \I_i)$ is a transversal matroid for $i \in [d-1]$.
Suppose that $\M_i$ is represented by a bipartite graph $G_i = (U, W_i; E_i)$ for $i \in [d-1]$. 
For $u \in U$ and $i \in [d-1]$, let $\delta_{G_i}(u) \subseteq E_i$ denote the set of edges incident to $u$ in $G_i$. 
Define $V \subseteq U \times E_1 \times \dots \times E_{d-1}$ as 
\[
V = \{(u, e_1, \dots , e_{d-1}) \mid u \in U,\ \mbox{$e_i \in \delta_{G_i}(u)$ for $i \in [d-1]$} \}. 
\]
For $v = (u, e_1, \dots , e_{d-1}) \in V$, we denote $\varphi_0(v) = u, \varphi_1(v) = e_1, \dots , \varphi_{d-1}(v) = e_{d-1}$. 

We define a matroid $\M'_0 = (V, \I'_0)$ as follows: 
a set $\{v_1, \dots , v_t \} \subseteq V$ is in $\I'_0$
if and only if $\varphi_0(v_1), \dots , \varphi_0(v_t)$ are distinct and $\{\varphi_0(v_1), \dots , \varphi_0(v_t)\} \in \I_0$. 
One can see that $\M'_0$ is a matroid, because 
it is obtained from a matroid $\M_0$ 
by splitting each element $u$ in $\M_0$ into multiple parallel copies. 
For $i \in [d-1]$, we define a matroid $\M'_i = (V, \I'_i)$ as follows: 
a set $\{v_1, \dots , v_t \} \subseteq V$ is in $\I'_i$
if and only if 
$\varphi_i(v_1), \dots , \varphi_i(v_t)$ do not share a vertex of $W_i$ in $G_i$. 
Note that $\M'_i$ is a simple partition matroid for $i \in [d-1]$. 
Let $\F'=\bigcap_{i=0}^{d-1}\I'_i$. 
Define a weight function $w'$ on $V$ as $w'(v) = w(\varphi_0(v))$ for $v \in V$. 

We now show the correspondence between $\F$ and $\F'$. 

\begin{claim}
For $X' \in \F'$, the set $\{\varphi_0(x') \mid x' \in X'\}$ is in $\F$. 
Conversely, for $X \in \F$, there exists $X' \in \F'$ such that 
$X = \{\varphi_0(x') \mid x' \in X'\}$.
\end{claim}

\begin{proof}[Proof of the claim]
For $X' \in \F'$, let $X = \{\varphi_0(x') \mid x' \in X'\} \subseteq U$. 
Then, $X$ is in $\I_0$ by the definition of $\I'_0$. 
Furthermore, for $i \in [d-1]$, the definition of $\I'_i$ shows that 
the set $\{\varphi_i(x') \mid x' \in X'\} \subseteq E_i$ forms a matching in $G_i$ that covers $X$. 
Therefore, $X$ is in $\F$. 

To show the latter half, let $X = (x_1, \dots , x_t) \in \F$. 
For $i \in [d-1]$, since $X \in \I_i$, $G_i$ has a matching $F_i = \{e_1^i, \dots , e_t^i \} \subseteq E_i$ that covers $X$, 
where we may assume that $e_j^i$ is incident to $x_j$. 
Then, $X' = \{ (x_j, e_j^1, \dots , e_j^{d-1}) \mid j \in [t]\}$ is a subset of $V$. 
By the construction, we have that $X' \in \F'$ and $X = \{\varphi_0(x') \mid x' \in X'\}$. 
This completes the proof. 
\end{proof}

Since $\M'_1, \dots , \M'_{d-1}$ are simple partition matroids,  
by Theorem~\ref{thm:partition}, we can compute a \weakcoreset $R' \subseteq V$ with respect to $\F', w'$, and $k$ of size at most $\tilde{O}(k^d)$. 
We now show that $R = \{ \varphi_0(r') \mid r' \in R' \}$ is a desired \weakcoreset with respect to $\F, w$, and $k$. 
Note that $|R| \le |R'| = \tilde{O}(k^d)$. 

Suppose we are given a set $X = \{x_1, \dots , x_t\} \in \F$ with $t \le k$ as in Definition~\ref{def:weakcoreset}. 
By the above claim, there exists $X' = \{x'_1, \dots , x'_t\} \in \F'$ such that $x_j = \varphi_0(x'_j)$ for $j \in [t]$. 
Since $R'$ is a \weakcoreset with respect to $\F', w'$, and $k$, 
$R'$ contains elements $y'_1, \dots , y'_t$ such that 
$w'(y'_j) \ge w'(x'_j)$ and $\{y'_1, \dots , y'_{j}, x'_{j+1}, \dots , x'_t\} \in \F'$ for any $j \in [t]$. 
Let $y_j = \varphi_0(y'_j)$ for $j \in [t]$. 
Then, $y_j \in R$ and $w(y_j) = w'(y'_j) \ge w'(x'_j) = w(x_j)$ for $j \in [t]$. 
Furthermore, for any $j \in [t]$, $\{y'_1, \dots , y'_{j}, x'_{j+1}, \dots , x'_t\} \in \F'$ implies 
$\{y_1, \dots , y_{j}, x_{j+1}, \dots , x_t\} \in \F$ by the above claim.
Thus, $R$ is a \weakcoreset with respect to $\F, w$, and $k$. 
This completes the proof of Theorem~\ref{thm:transversal}. 
\end{proof}

Since a (non-simple) partition matroid is a special case of a transversal matroid, 
a similar result holds also for partition matroids. 
We say that a matroid $\mathcal{M} = (E, \mathcal{I})$ is a {\em partition matroid} if 
there exist a partition $(E^1, \ldots, E^{\ell})$ of $E$ and a function $c: [\ell] \to \mathbb{Z}_{\geq 0}$ such that
\begin{equation*}
\I = \{ I \subseteq E \mid |I \cap E^j| \leq c(j) \text{ for any } j \in [\ell] \}.
\end{equation*}
A partition matroid is a transversal matroid 
represented by a bipartite graph in which all vertices in each $E^j$ are adjacent 
to the same $c(j)$ common neighboring vertices, for every $j \in [\ell]$.
Therefore, we obtain the following result as a corollary of Theorem~\ref{thm:transversal}.

\begin{corollary}
For the weighted $d$-matroid intersection problem such that, out of the $d$ given matroids, all except one are partition matroids,
we can construct a \weakcoreset of size $\tilde{O}(k^{d})$ in polynomial time.
\end{corollary}

%% file: laminar.tex
\section{Kernelization Algorithm for Laminar Matroids}
\label{sec:laminar}

In this section, we prove Theorem~\ref{thm:laminar}. 
We may assume that $|E| = k^{\Omega(\log k)}$, since otherwise $E$ is trivially a \weakcoreset of size $k^{O(\log k)}$.

Let $E$ be a finite set. 
A set family $\mathcal{L} \subseteq 2^E$ is called a {\em laminar family} over $E$ if 
for every two sets $A$ and $B$ in $\mathcal{L}$ with $A \cap B \neq \emptyset$, 
either $A \subseteq B$ or $B \subseteq A$ holds. 
A matroid $\mathcal{M} = (E, \mathcal{I})$ is called a {\em laminar matroid} if 
there exist a laminar family $\mathcal{L}$ over $E$ and a function $c \colon \mathcal{L} \to \mathbb{Z}_{\ge 0}$
such that 
$$
\mathcal{I} = \{ I \subseteq E \mid |I \cap A| \le c(A) \mbox{ for any } A \in \mathcal{L} \}. 
$$

Recall that, for a matroid $\mathcal{M} = (E, \mathcal{I})$ and for $X \in \mathcal{I}$ and $x \in X$, 
$Y \subseteq E$ is said to be an $(X, x)$-exchangeable set if 
$x \in Y$, and $X -x +y \in \mathcal{I}$ for every $y \in Y$. 
To prove Theorem~\ref{thm:laminar}, we show the following lemma.  

\begin{lemma}
\label{lem:laminarcandidate}
Let $\mathcal{M} = (E, \mathcal{I})$ be a laminar matroid and let $k$ be a positive integer. 
There exists a randomized algorithm that runs in polynomial time and 
returns $\mathcal{Y} \subseteq 2^E$ with $|\mathcal{Y}| = k^{O(\log k)}$ 
such that 
for any $X \in \mathcal{I}$ with $|X| \le k$ and for any $x \in X$, 
$\mathcal{Y}$ contains an $(X, x)$-exchangeable set with probability at least $\frac{1}{k^{O(\log k)}}$.  
\end{lemma}

\begin{proof}
    Suppose that $\mathcal{M}$ is represented by a laminar family $\mathcal{L} \subseteq 2^E$ and 
    a function $c \colon \mathcal{L} \to \mathbb{Z}_{\ge 0}$. 
    Since we only consider independent sets of size at most $k$, 
    we may assume that $E \in \mathcal{L}$ and $c(E) = k$. 
    Since we assume that the matroid contains no loop, it holds that $c(A)>0$ for any $A \in \mathcal{L}$. 
    Furthermore, we may assume that $\{e\} \in \mathcal{L}$ and $c(\{e\}) =1$ for every element $e \in E$, 
    because $\mathcal{L}$ remains a laminar family even after adding $\{e\}$. 
    We may also assume that if $A \subseteq B$ for $A, B \in \mathcal{L}$, then $c(A) \le c(B)$, 
    since otherwise, we can remove $A$ from $\mathcal{L}$. 

    Let $\mathcal{Z} \subseteq \mathcal{L}$ be a disjoint set family, i.e., a set family whose members are mutually disjoint. 
    We denote 
    \[
        {\rm height}(\mathcal{Z}) = \max \{c(A) \mid A \in \mathcal{Z}\}, \qquad 
    E(\mathcal{Z}) = \bigcup_{Z \in \mathcal{Z}} Z.
    \]    
    For a positive integer $i \le {\rm height}(\mathcal{Z})$, let $\mathcal{Z}_i$ be the collection of 
    inclusionwise maximal sets in $\{A \in \mathcal{L} \mid \mbox{$\exists Z \in \mathcal{Z}$ with $A \subseteq Z$,\ $c(A) \le i$} \}$, 
    which forms a refined partition of $E(\mathcal{Z})$.

    To show the lemma, we consider a recursive algorithm \texttt{FindCandidate}, 
    whose input is a disjoint set family $\mathcal{Z} \subseteq \mathcal{L}$. 
    Let $h = {\rm height}(\mathcal{Z})$. 
    When $h = 1$, the algorithm simply outputs $\{E(\mathcal{Z}) \}$.  
    When $h \ge 2$, the algorithm performs the following procedure for each 
    $i \in \{\lceil \frac{h}{2} \rceil, \dots, h-1\}$ and outputs the union of the set families obtained.
    \begin{itemize}
    \item
     Sample each set in $\mathcal{Z}_{i}$ independently with probability $\frac{1}{k}$ to obtain $\mathcal{Z}'$. 
    \item
     Sample each set in $\mathcal{Z}'_{h-i}$ independently with probability $\frac{1}{k}$ to obtain $\mathcal{Z}''$. 
    \item
    Execute \texttt{FindCandidate} for $\mathcal{Z}''$ to obtain a set family.
    \end{itemize}
    The pseudocode for this algorithm is presented in Algorithm~\ref{alg:findcandidate}.

\begin{algorithm} 
    $h \gets {\rm height}(\mathcal{Z})$; \\ 
    \If{$h=1$} { 
        \Return $\{E(\mathcal{Z})\}$\;
    }
    \Else{
        $\mathcal{Y} \gets \emptyset$\; 
        \For{$i=\lceil \frac{h}{2} \rceil$ \textup{to} $h-1$}{
            $\mathcal{Z}', \mathcal{Z}'' \gets \emptyset$\;  
            \For{\textup{each} $Z \in \mathcal{Z}_{i}$}
                {
                $\mathcal{Z}' \gets \mathcal{Z}' \cup \{Z\}$ with probability $\frac{1}{k}$\;
                }
            \For{\textup{each} $Z \in \mathcal{Z}'_{h-i}$}
                {
                $\mathcal{Z}'' \gets \mathcal{Z}'' \cup \{Z\}$ with probability $\frac{1}{k}$\;
                }
            $\mathcal{Y} \gets \mathcal{Y} \cup \texttt{FindCandidate} (\mathcal{Z}'')$\;  
        }
        \Return $\mathcal{Y}$\;
    }    
    \caption{\texttt{FindCandidate}$(\mathcal{Z})$}\label{alg:findcandidate}
\end{algorithm}
   
    In what follows, we show that \texttt{FindCandidate}$(\{E\})$ returns a set family satisfying the conditions in the lemma. 

    \begin{claim}
        For a disjoint set family $\mathcal{Z} \subseteq \mathcal{L}$ with ${\rm height}(\mathcal{Z}) = h$, 
        {\rm \texttt{FindCandidate}}$(\mathcal{Z})$ returns a set family of size at most $k^{\log_2 h}$. 
    \end{claim}

    \begin{proof}[Proof of the claim]
    We prove the claim by induction on $h$. 
    If $h = 1$, then the size of the output of \texttt{FindCandidate}$(\mathcal{Z})$ is $1$, which is equal to $k^{\log_2 1}$.      
    Suppose that $h \ge 2$. 
    Then, for $i=\lceil \frac{h}{2} \rceil, \dots , h-1$, 
    \texttt{FindCandidate}$(\mathcal{Z})$ calls \texttt{FindCandidate}$(\mathcal{Z}'')$, where ${\rm height}(\mathcal{Z}'') \le h-i$. 
    Therefore, the induction hypothesis shows that the size of the output of \texttt{FindCandidate}$(\mathcal{Z})$ is at most 
    \[
       \sum_{i=\lceil \frac{h}{2} \rceil}^{h-1}  k^{\log_2 h-i} \le \sum_{i=\lceil \frac{h}{2} \rceil}^{h-1} k^{\log_2 \frac{h}{2}} \le \frac{h}{2} \cdot k^{(\log_2 h) -1} \le k \cdot k^{(\log_2 h) -1} = k^{\log_2 h}, 
    \]
    which completes the proof. 
   \end{proof}
    This implies that \texttt{FindCandidate}$(\{E\})$ returns a set family of size at most $k^{O(\log k)}$.

    We fix $X \in \mathcal{I}$ with $|X| \le k$ and $x \in X$. 
    A set $A \in \mathcal{L}$ is said to be {\em tight} if $|(X - x) \cap A| = c(A)$. 
    Note that if an element $y \in E$ is contained in no tight set, then $X-x+y \in \I$. 
    For a disjoint set family $\mathcal{Z} \subseteq \mathcal{L}$, we say that $\mathcal{Z}$ is {\em good} if 
    the following conditions hold. 
    \begin{enumerate}
        \item There exists a set $Z_0 \in \mathcal{Z}$ such that $x \in Z_0$ and $X \cap (E(\mathcal{Z}) \setminus Z_0) = \emptyset$. 
        \item For any $Z \in \mathcal{Z}$, there exists no tight set containing $Z$. 
    \end{enumerate}
    Note that the first condition implies that each $Z \in \mathcal{Z}$ is not tight.  
    Indeed, if $Z = Z_0$, then $|(X-x) \cap Z| = |X \cap Z| - 1 \le c(Z) -1$; and 
    if $Z \in \mathcal{Z} \setminus \{Z_0\}$, then $|(X-x) \cap Z| = 0 < c(Z)$. 
    Therefore, the second condition is equivalent to the nonexistence of a tight set that strictly contains some $Z \in \mathcal{Z}$. 
    The set $Z_0$ that satisfies the first condition is called an {\em $x$-containing set}. 
    For a good set family $\mathcal{Z}$ with the $x$-containing set $Z_0$, define 
    \[
    \phi(\mathcal{Z}) = \max\{c(A) \mid \mbox{$A \in \mathcal{L}$ is a tight set contained in $Z_0$} \},  
    \]
    where we denote $\phi(\mathcal{Z}) = - \infty$ if no such $A$ exists. 
    Since $Z_0$ is not tight as observed above, we obtain $\phi(\mathcal{Z}) \le c(Z_0)-1 \le {\rm height}(\mathcal{Z}) -1$.

   \begin{claim}
        Let $\mathcal{Z} \subseteq \mathcal{L}$ be a good set family 
        and let $i$ be a positive integer with $\phi(\mathcal{Z}) \le i \le {\rm height}(\mathcal{Z})$. 
        Let $\mathcal{Z}'$ be a set family obtained by sampling 
        each set in $\mathcal{Z}_{i}$ independently with probability $\frac{1}{k}$ (as in lines 8--9 in Algorithm~\ref{alg:findcandidate}). 
        Then, $\mathcal{Z}'$ is good with probability at least $\Omega(\frac{1}{k})$.
   \end{claim}

    \begin{proof}[Proof of the claim]        
        Observe that $\mathcal{Z}_i$ contains a unique set $Z'_0$ with $x \in Z'_0$ and 
        at most $k-1$ other sets intersecting with $X$. 
        We see that $\mathcal{Z}'$ satisfies the first condition of a good set family if and only if
        $Z'_0$ is added to $\mathcal{Z}'$ and no other sets intersecting with $X$ are added to $\mathcal{Z}'$. 
        This happens with probability at least 
        \[
            \frac{1}{k} \cdot \left(1 - \frac{1}{k}\right)^{k-1} = \Omega \left( \frac{1}{k} \right). 
        \]

        We show that $\mathcal{Z}'$ satisfies the second condition of a good set family, 
        assuming that it satisfies the first condition.
        Fix $Z' \in \mathcal{Z}'$ and  
        let $Z \in \mathcal{Z}$ be the set with $Z' \subseteq Z$. 
        Since $\mathcal{L}$ is a laminar family, a set $A \in \mathcal{L}$ containing $Z'$ satisfies that 
        $A = Z'$, $Z' \subsetneq A \subsetneq Z$, or $Z \subseteq A$. 
        Then, we can show that there exists no tight set containing $Z'$ as follows. 
        \begin{itemize}
            \item  Since $\mathcal{Z}'$ satisfies the first condition of a good set family, $Z'$ is not tight. 
            \item  Any set $A \in \mathcal{L}$ with $Z' \subsetneq A \subsetneq Z$ is not tight as follows: 
                \begin{itemize}
                    \item if $Z$ is not the $x$-containing set of $\mathcal{Z}$, 
                    then $X \cap Z = \emptyset$, and hence $|(X-x) \cap A| = 0 < c(A)$, and  
                    \item if $Z$ is the $x$-containing set of $\mathcal{Z}$, then $c(A) > i \ge \phi(\mathcal{Z})$ by the definition of $\mathcal{Z}_i$, and hence $A$ is not tight by the maximality of $\phi({\mathcal{Z}})$. 
                \end{itemize}
            \item Since $\mathcal{Z}$ is good, there exists no tight set containing $Z$. 
        \end{itemize}
        Thus, $\mathcal{Z}'$ satisfies the second condition. 

        Therefore, $\mathcal{Z}'$ is good with probability at least $\Omega(\frac{1}{k})$. 
    \end{proof}

    \begin{claim}
        Suppose that $\mathcal{Z} \subseteq \mathcal{L}$ is a good set family with ${\rm height}(\mathcal{Z}) = h \ge 2$. 
        Then, there exists $i \in \{ \lceil \frac{h}{2} \rceil, \dots , h-1\}$ such that 
        the set family $\mathcal{Z}''$ obtained by applying Lines 7--11 in Algorithm~\ref{alg:findcandidate} is good 
        with probability at least $\Omega(\frac{1}{k^2})$. 
    \end{claim}

    \begin{proof}[Proof of the claim]  
        We first consider the case where $\phi(\mathcal{Z}) < \lceil \frac{h}{2} \rceil$. 
        Let $\mathcal{Z}'$ and $\mathcal{Z}''$ be the set families
        obtained by applying Lines 7--11 in Algorithm~\ref{alg:findcandidate} for $i= \lceil \frac{h}{2} \rceil$. 
        Since $\phi(\mathcal{Z}) \le i$, the previous claim shows that $\mathcal{Z}'$ is good with probability at least $\Omega(\frac{1}{k})$. 
        Under the assumption that $\mathcal{Z}'$ is good, 
        since $\phi(\mathcal{Z}') \le \phi(\mathcal{Z}) \le \lceil \frac{h}{2} \rceil - 1 \le h - i$, 
        the previous claim shows that $\mathcal{Z}''$ is good with probability at least $\Omega(\frac{1}{k})$.
        By combining them, $\mathcal{Z}''$ is good with probability at least $\Omega(\frac{1}{k^2})$.

        We next consider the case where $\phi(\mathcal{Z}) \ge \lceil \frac{h}{2} \rceil$. 
        Let $\mathcal{Z}'$ and $\mathcal{Z}''$ be the set families
        obtained by applying Lines 7--11 in Algorithm~\ref{alg:findcandidate} for $i= \phi(\mathcal{Z})$. 
        Then, the previous claim shows that $\mathcal{Z}'$ is good with probability at least $\Omega(\frac{1}{k})$. 
        Suppose that $\mathcal{Z}'$ is good and let $Z'_0 \in \mathcal{Z}'$ be the $x$-containing set of $\mathcal{Z}'$. 
        By the definition of $\phi(\mathcal{Z})$, there exists a tight set $A^* \in \mathcal{L}$ such that $c(A^*) = \phi(\mathcal{Z})$ and $A^* \subseteq Z_0$, where $Z_0 \in \mathcal{Z}$ is the $x$-containing set of $\mathcal{Z}$. 
        By taking an inclusionwise maximal set among such sets $A^*$, we may assume that $A^* \in \mathcal{Z}_i$.
        Since $Z'_0$ is also a member of $\mathcal{Z}_i$, $Z'_0$ and $A^*$ are disjoint sets. 
        Note that $Z'_0 \neq A^*$ follows from the fact that $A^*$ is tight, while $Z'_0$ is not tight.  
        Therefore, 
        \[
            |(X-x) \cap Z'_0| \le |(X-x) \cap Z_0| - |(X-x) \cap A^*| \le c(Z_0) - c(A^*) \le h - \phi(\mathcal{Z}). 
        \]
        Then, for any tight set $A$ contained in $Z'_0$, we obtain
        \[
        c(A) = |(X-x) \cap A| \le |(X-x) \cap Z'_0| \le h-\phi(\mathcal{Z}).
        \]
        Since this implies that $\phi(\mathcal{Z}') \le h-\phi(\mathcal{Z}) = h-i$,  
        the previous claim shows that $\mathcal{Z}''$ is good with probability at least $\Omega(\frac{1}{k})$.
        By combining them, $\mathcal{Z}''$ is good with probability at least $\Omega(\frac{1}{k^2})$.
    \end{proof}

   \begin{claim}
        For a good set family $\mathcal{Z} \subseteq \mathcal{L}$ with ${\rm height}(\mathcal{Z}) = h$, 
        {\rm \texttt{FindCandidate}}$(\mathcal{Z})$ returns a set family 
		that contains an $(X, x)$-exchangeable set with probability at least $\frac{1}{k^{2 \log_2 h}}$.  
    \end{claim}

    \begin{proof}[Proof of the claim]
        We prove the claim by induction on $h$. 
        Suppose that $h=1$. 
        In this case, for any $Z \in \mathcal{Z}$ and for any $y \in Z$, we see that there exists no tight set containing $y$ as follows: for any $A \in \mathcal{L}$ with $Z \subseteq A$, $A$ is no tight as $\mathcal{Z}$ is good; and for any $A \in \mathcal{L}$ with $y \in A \subsetneq Z$, 
        $|(X-x) \cap A| \le |(X-x) \cap Z| < c(Z) = 1 = c(A)$. 
        Therefore, for any $y \in E(\mathcal{Z})$, we obtain $X - x + y \in \mathcal{I}$. 
        We also see that $x \in E(\mathcal{Z})$, 
        because $\mathcal{Z}$ is good. 
        Therefore, $E(\mathcal{Z})$ is an $(X, x)$-exchangeable set. 
        This shows that the output $\{E(\mathcal{Z})\}$ of \texttt{FindCandidate}$(\mathcal{Z})$ is a set family that contains an $(X, x)$-exchangeable set with probability $1$. This probability is written as $\frac{1}{k^{2 \log_2 h}}$.

        Suppose next that $h \ge 2$. In this case, by the previous claim, 
        there exists $i \in \{ \lceil \frac{h}{2} \rceil, \dots , h-1\}$ such that 
        the set family $\mathcal{Z}''$ obtained by applying Lines 7--11 in Algorithm~\ref{alg:findcandidate} is good 
        with probability at least $\Omega(\frac{1}{k^2})$. 
        Under the assumption that such a set family $\mathcal{Z}''$ is good, since ${\rm height}(\mathcal{Z}'') \le \frac{h}{2}$,  
        the output of \texttt{FindCandidate}$(\mathcal{Z}'')$ 
        contains an $(X, x)$-exchangeable set with probability at least $\frac{1}{k^{2 \log_2 (h/2)}}$ by the induction hypothesis.
        Since the output of \texttt{FindCandidate}$(\mathcal{Z})$ contains that of \texttt{FindCandidate}$(\mathcal{Z}'')$, 
        it contains an $(X, x)$-exchangeable set with probability at least 
        \[
        \frac{1}{k^2} \cdot \frac{1}{k^{2 \log_2 (h/2)}} = \frac{1}{k^{2 \log_2 h}}.
        \]
        This completes the proof of the claim. 
    \end{proof}

    Since $\{E\} \subseteq \mathcal{L}$ is a good set family with ${\rm height}(\{E\}) = k$, by the above claims, 
    \texttt{FindCandidate}$(\{E\})$
    returns $\mathcal{Y} \subseteq 2^E$ with $|\mathcal{Y}| = k^{O(\log k)}$ such that 
    $\mathcal{Y}$ contains an $(X, x)$-exchangeable set with probability at least $\frac{1}{k^{O(\log k)}}$.  
    This completes the proof of Lemma~\ref{lem:laminarcandidate}.       
    \end{proof}

    We are now ready to prove Theorem~\ref{thm:laminar}, which we restate here. 

    \laminar*

    \begin{proof}
        For $i\in [d-1]$, let $\mathcal{M}_i = (E, \mathcal{I}_i)$ be a laminar matroids, and 
        let $\mathcal{M}_0 = (E, \mathcal{I}_0)$ be an arbitrary matroid. 
        Suppose that $\mathcal{F} = \bigcap_{i=0}^{d-1}\I_i$. 
        
        For $i \in [d-1]$, we apply Lemma~\ref{lem:laminarcandidate} to a laminar matroid $\mathcal{M}_i$ 
        to obtain a set family $\mathcal{Y}_i \subseteq 2^E$. 
        Using $\mathcal{Y}_1, \dots , \mathcal{Y}_{d-1}$, we construct the set family
        \[
            \mathcal{Y} = \bigg\{ \bigcap_{i=1}^{d-1} Y_i \mid Y_i \in \mathcal{Y}_i \mbox{ for } i \in [d-1] \bigg\}.         
        \]
        Then, for each $Y \in \mathcal{Y}$, we apply $\texttt{Greedy}(\M_0, Y, k)$ and add its output to $R$. 
        We repeat this process $T = k^{\Theta(d \log k)}$ times and returns the obtained set $R$. 
        The pseudocode for the algorithm is shown in Algorithm~\ref{alg:laminarkernel}. 

    \begin{algorithm} 
    $R \gets \emptyset$\; 
    \RepTimes{$T$}{
        \For{$i=1$ \textup{to} $d-1$}{
            Apply Lemma~\ref{lem:laminarcandidate} to $\mathcal{M}_i$ to obtain $\mathcal{Y}_i$\; 
        }
        $\mathcal{Y} \gets \big\{ \bigcap_{i=1}^{d-1} Y_i \mid Y_i \in \mathcal{Y}_i \mbox{ for } i \in [d-1] \big\}$\; 
        \For{\textup{each} $Y \in \mathcal{Y}$}
            {
                $R' \gets \texttt{Greedy}(\M_0, Y, k)$\; 
                $R \gets R \cup R'$\; 
            }
     }
    \Return $R$\;
    \caption{\texttt{LaminarKernel}$(\mathcal{M}_0, \dots , \mathcal{M}_{d-1})$}\label{alg:laminarkernel}
    \end{algorithm}

        We first show that Algorithm~\ref{alg:laminarkernel} returns a set $R$ with $|R| = k^{O(d \log k)}$. 
        Since each $\mathcal{Y}_i$ contains $k^{O(\log k)}$ sets, $\mathcal{Y}$ contains at most $\prod_{i=1}^{d-1} |\mathcal{Y}_i| = k^{O(d \log k)}$ sets. 
        Since $\texttt{Greedy}(\M_0, Y, k)$ returns a set of size at most $k$, 
        one execution of lines 3--8, called a {\em round}, adds at most $k \cdot k^{O(d \log k)}$ elements to $R$. 
        Since we execute $T = k^{\Theta(d \log k)}$ rounds, the output $R$ satisfies that 
        $|R| = k \cdot k^{O(d \log k)} \cdot k^{\Theta(d \log k)}$, which is rewritten as $|R| = k^{O(d \log k)}$. 

        We next show that the output $R$ satisfies (SingleEXC).  
        Fix $X \in \mathcal{F}$ with $|X| \le k$ and $x \in X$. 
        We say that a round is {\em successful} 
        if $\mathcal{Y}_i$ contains an $(X, x)$-exchangeable set with respect to $\mathcal{M}_i$ for every $i \in [d-1]$.     
        By Lemma~\ref{lem:laminarcandidate}, a round is successful with probability at least $\frac{1}{k^{O(d \log k)}}$.     
        Therefore, if we execute $k^{\Theta(d \log k)}$ rounds (with an appropriately large hidden constant), 
        then there exists at least one successful round with probability at least $1 - \frac{1}{3k}$.

        We now show that if there exists at least one successful round, then 
        there exists an element $y$ satisfying the conditions in (SingleEXC). 
        Consider $\mathcal{Y}_1, \dots , \mathcal{Y}_{d-1}$ in a successful round. 
        For $i \in [d-1]$, let $Y^*_i \in \mathcal{Y}_i$ be an $(X, x)$-exchangeable set with respect to $\mathcal{M}_i$.   
        Then, $\mathcal{Y}$ contains $Y^* := \bigcap_{i=1}^{d-1} Y^*_i$. 
        By Lemma~\ref{lem:exchangeablesetgreedy}, the output $R^*$ of $\texttt{Greedy}(\M_0, Y^*, k)$ contains an element $y$
        such that $y \notin X-x$, $w(y) \ge w(x)$, and $X -x + y \in \mathcal{F}$.
        Since $R^*$ is a subset of the output $R$ of Algorithm~\ref{alg:laminarkernel}, $y$ is a desired element in $R$. 

        Since this argument works for any $X$ and $x$, the output $R$ satisfies (SingleEXC), 
        which shows that it is a \weakcoreset. 
        Moreover, the algorithm runs in polynomial time in $|E|$, because we assume that $|E| = k^{\Omega(\log k)}$. 
\end{proof}

%% file: matroidalmatching.tex
\section{Matching with a Matroid Constraint} \label{sec:matroidal_matching}

Recall that in the weighted matching problem with a matroid constraint, 
we are given a graph $G = (V, E)$, a matroid $\M = (E, \I)$ defined on the edge set of $G$, and a weight function $w: E\to \mathbb{R}_{\ge 0}$. 
The goal of the problem is to find a matching of maximum weight that is independent in $\M$. 
Let $\F = \{X \subseteq E \mid X \text{ is a matching and } X \in \I \}$.
Let $k$ be a parameter that represents an upper bound on the solution size $|X|$. 
In this section, we prove Theorem~\ref{thm:matroidalmatching}, which we restate here. 

\matroidalmatching*

\begin{proof}

Our algorithm for computing a \weakcoreset for the weighted matching problem with a matroid constraint is based on the same ideas as 
those in Section~\ref{subsec:partition}.

\begin{algorithm}
    \caption{\texttt{MatroidalMatchingKernel}$(G = (V, E), \M = (E, \I))$}\label{alg:matroidalmatchingkernel}
    $R \gets \emptyset$; \\
    \RepTimes{$T$}{
        Let $S \subseteq V$ be a subset obtained by independently sampling each vertex $v \in V$ with probability $\frac{1}{2k}$;  \\
        $R' \gets \texttt{Greedy}(\M, E(S) , k)$; \\
        \tcp{$E(S)$ is the set of edges whose both endpoints are in $S$.}
        $R \gets R \cup R'$; \\
    }
    \Return $R$;
\end{algorithm}

In our algorithm, we repeat the following procedure $T = \tilde{O}(k^2)$ times. 
We sample each vertex $v \in V$ independently with probability $\frac{1}{2k}$ to obtain a random subset $S\subseteq V$.
Then we apply $\texttt{Greedy}(\M, E(S), k)$, where $E(S)$ denotes the set of edges whose both endpoints are contained in $S$.
We denote the resulting set by $R'$ and add it to our \weakcoreset $R$.
The pseudocode for this algorithm is shown in Algorithm~\ref{alg:matroidalmatchingkernel}.

For a graph $G = (V, E)$ and for a matching $X \subseteq E$ in $G$ and $x \in X$, 
we say that $Y \subseteq E$ is an {\em $(X, x)$-exchangeable set with respect to $G$} if 
$x \in Y$ and $X - x + y$ is a matching in $G$ for every $y \in Y$. 

We will prove that the obtained set $R$ satisfies~(SingleEXC).
Let $X \in \mathcal{F}$ be a feasible solution of size at most $k$, and $x$ be an arbitrary element in $X$.

Similarly to the proof of Theorem~\ref{thm:partition}, consider one round in the algorithm where we generate the sets $S$ by sampling.
We say that the round is {\em successful} if 
$E(S)$ is an $(X, x)$-exchangeable set with respect to $G$.
We now evaluate the probability that one round is successful. 

Let $u$ and $v$ be the endpoints of $x$. 
Observe that 
$E(S)$ is an $(X, x)$-exchangeable set with respect to $G$  
if 
$u \in S$, $v \in S$, and the set $S$ contains no endpoint of any edge $e \in X - x $. 
Hence, the probability that $E(S)$ is an $(X, x)$-exchangeable set is at least 
\begin{equation*}
\frac{1}{(2k)^2} \left( 1 -  \frac{1}{2k} \right)^{2k - 2} \geq 
\frac{1}{4ek^2},
\end{equation*}
where we use the fact that $(1 - 1/\alpha)^{\alpha - 2} \geq 1/e$ for $\alpha > 1$. 
Therefore, the probability that one round is successful is at least $\frac{1}{4ek^2}  $.
By setting $T = \lceil 4e k^2 \log(3k) \rceil$,  
the probability that there exists no successful round among the $T$ rounds is at most 
\begin{equation*}
\begin{aligned}
&\left( 1 -  \frac{1}{4e k^2}  \right)^T \leq \frac{1}{3k}, 
\end{aligned}
\end{equation*}
where we use $(1 - 1/\alpha)^{\alpha} \leq 1/e$ for $\alpha \ge 1$. 
Therefore, with probability at least $1 - \frac{1}{3k}$, there exists at least one successful round.

In the same way as the proof of Theorem~\ref{thm:partition}, we show that $R$ satisfies (SingleEXC).
Let $R'$ be the output of $\texttt{Greedy}(\M_0, E(S), k)$ for the sets $S$ generated in a successful round. 
Since $x \in E(S)$, by Lemma~\ref{lem:greedyproperty}, there exists $y \in R' \setminus (X-x)$ such that $w(y) \geq w(x)$ and $X - x + y \in \I$.
Now, since $E(S)$ is an $(X, x)$-exchangeable set with respect to $G$, we have $X - x + y$ is a matching.
Thus, in a successful round, there exists $y \in R' \setminus (X-x)$ such that $w(y) \geq w(x)$ and $X - x + y \in \F$.
Therefore, $R$ satisfies (SingleEXC).

Moreover, Algorithm~\ref{alg:matroidalmatchingkernel} runs in polynomial time, and the size of the set $R$ returned by Algorithm~\ref{alg:matroidalmatchingkernel} is at most 
$k \cdot T = O\left(k^3 \log k \right)$.
\end{proof}

%% file: deterministic.tex
\section{Deterministic Reachable Kernel}
\label{sec:deterministic}

In this section, we discuss deterministic algorithms for finding a reachable kernel. 
Recall that a subset $R \subseteq E$ is a {\rm \coreset} if 
it satisfies the conditions in Definition~\ref{def:weakcoreset} with probability $1$. 
Equivalently, $R \subseteq E$ is a {\rm \coreset} if it always satisfies the conditions in (SingleEXC), 
which is described as follows. 

\begin{description}
   \item[(SingleEXC)']
    For any $X \in \mathcal{F}$ with $|X| \le k$ and any $x \in X$, there exists $y \in R \setminus (X-x)$ such that 
    $w(y) \ge w(x)$ and $X - x + y \in \mathcal{F}$.
\end{description}

We now prove Theorem~\ref{thm:deterministic}, which we restate here. 

\deterministic*

\begin{algorithm}
 \caption{\texttt{DeterministicKernel}$(\mathcal{M}_0, \mathcal{M}_1)$}\label{alg:strongkernel}
    $R \gets \emptyset$\;
    $U \gets E$\; 
    \RepTimes{$g(k)+1$}{
        $\M_i \gets \M_i | U~(i=0,1)$\; 
        $F \gets \texttt{Greedy} (\M_0, U, k)$\;
        $R\gets R\cup F$\;
        \For{{\rm each} $e\in F$}{
          $H\gets \texttt{Greedy} (\M_0, \textrm{span}_{\M_1} (e), k)$\;
          $R \gets R\cup H$\;
        }
        $U\gets E\setminus \bigcup_{e\in F}\textrm{span}_{\M_1} (e)$\;
      }
    \Return $R$\;
\end{algorithm}

\begin{proof}
Let $\M_0=(E, \I_0)$ be an arbitrary matroid and $\M_1=(E, \I_1)$ be a \gkgood matroid.
We may assume that the rank of $\M_0$ is equal to $k$ by applying truncation if necessary.
Let $\F=\I_0\cap \I_1$.
We show that the output $R$ of Algorithm~\ref{alg:strongkernel} is a \coreset of size at most $k^2 (g(k)+1)$. 

We denote the sets $F$, $H$ obtained in the $t$-th iteration by $F_t$, $H_t$, respectively.
Then $R=\bigcup_{t} (F_t\cup H_t)$.
Also, we denote $U_t=E\setminus \bigcup_{t' < t} \bigcup_{e\in F_{t'}}\textrm{span}_{\M_1} (e)$, which is the set of remaining elements at the beginning of the $t$-th iteration.

To show that $R$ satisfies (SingleEXC)', let $X$ be a feasible solution of size at most $k$, and $x$ be an element in $X$.
We will show that, for some $t$, there exists $y\in (F_t\cup H_t) \setminus (X-x)$ such that $w(y)\geq w(x)$ and $X-x+y\in \F$.

First suppose that there exists $e\in F_t$ such that $x\in \textrm{span}_{\M_1}(e)$.
Then, since $\textrm{span}_{\M_1}(e)$ is an $(X, x)$-exchangeable set with respect to $\M_1$, 
it follows from Lemma~\ref{lem:exchangeablesetgreedy} that 
the output $H_t$ of $\texttt{Greedy}(\M_0, \textrm{span}_{\M_1}(e), k)$ contains an element $y$ such that  
$y \notin X-x$, $w(y) \ge w(x)$, and $X-x+y \in \F$.
Since $H_t \subseteq R$, this means that $y$ satisfies the conditions in (SingleEXC)'.

In what follows, we assume that, for any $t$, it holds that $x\not\in \textrm{span}_{\M_1}(e)$ for any $e\in F_t$.
Then, $x \in U_t$ for any $t$. 
Since $F_t=\texttt{Greedy}(\M_0, U_t, k)$,
it follows from Lemma~\ref{lem:greedyproperty} that there exists an element $y_t \in F_t \setminus (X-x)$ such that $w(y_t) \ge w(x)$ and $X-x+y_t \in \mathcal{I}_0$.     
Since $U_t=U_{t-1}\setminus \bigcup_{e\in F_t}\textrm{span}_{\M_1} (e)$, we have $\textrm{span}_{\M_1}(y_t)\neq \textrm{span}_{\M_1}(y_{t'})$ for any $t \neq t'$.
On the other hand, since $\M_1$ is \gkgood, there exists $Y\subseteq {E}$ such that $|Y|\leq g(k)$ and $\textrm{span}_{\M_1}(X-x)=\bigcup_{y \in Y} \mathrm{span}_{\M_1}(y)$.
Therefore, there exists at least one $y_t$ such that $y_t\not\in \textrm{span}_{\M_1}(X-x)$.
Hence we have $X-x+y_t\in \I_1$, meaning that $X-x+y_t\in \I_0\cap \I_1$ and $w(y_t) \ge w(x)$.
Thus, $y$ satisfies the conditions in (SingleEXC)'. 

Moreover, it is clear that the algorithm runs in polynomial time and that the size of $R$ is at most $k^2 (g(k)+1)$. \end{proof}

%% file: matroidclass.tex
\section{Matroid Classes}
\label{sec:matroidclass}

In this section, we list the matroid classes dealt with in this paper.

\paragraph{Uniform matroids.}
A matroid $\mathcal{M} = (E, \mathcal{I})$ is called a {\em uniform matroid} if 
there exists a positive integer $c$ such that $\I$ consists of all sets of size at most $c$. 

\paragraph{Partition matroids.}
A matroid $\mathcal{M} = (E, \mathcal{I})$ is called a {\em partition matroid} if 
there exist a partition $(E^1, \ldots, E^{\ell})$ of $E$ and a function $c: [\ell] \to \mathbb{Z}_{\geq 0}$ such that
\begin{equation*}
\I = \{ I \subseteq E \mid |I \cap E^j| \leq c(j) \text{ for any } j \in [\ell] \}.
\end{equation*}
If $c(j) = 1$ for any $j \in [\ell]$, then it is called a {\em simple partition matroid}. 
Uniform matroids are partition matroids.

\paragraph{Graphic matroids.}
A matroid $\mathcal{M} = (E, \mathcal{I})$ is called a {\em graphic matroid} if 
there exists a graph $G = (V, E)$ such that $F \subseteq E$ is an independent set of $\M$ if and only if $F$ forms a forest in $G$. 
Simple partition matroids are graphic where a graph consists of independent parallel edges.

\paragraph{Cographic matroids.}
A matroid $\M = (E, \I)$ is called a {\em cographic matroid}
if there exists a graph $G=(V, E)$ such that $F \subseteq E$ belongs to $\I$ 
if and only if $G - F$ has the same number of connected components as $G$. 

\paragraph{Transversal matroids.}
Let $G = (U, W; E)$ be a bipartite graph with a bipartition $(U, W)$. 
The {\em transversal matroid} represented by $G$ is a matroid $\M = (U, \I)$ over $U$ such that 
$X \subseteq U$ is in $\I$ if and only if $G$ has a matching that covers $X$. 
Partition matroids are known to be transversal.

\paragraph{Laminar matroids.}
Let $E$ be a finite set. 
A set family $\mathcal{L} \subseteq 2^E$ is called a {\em laminar family} over $E$ if 
for every two sets $A$ and $B$ in $\mathcal{L}$ with $A \cap B \neq \emptyset$, 
either $A \subseteq B$ or $B \subseteq A$ holds. 
A matroid $\mathcal{M} = (E, \mathcal{I})$ is called a {\em laminar matroid} if 
there exist a laminar family $\mathcal{L}$ over $E$ and a function $c \colon \mathcal{L} \to \mathbb{Z}_{\ge 0}$
such that 
$$
\mathcal{I} = \{ I \subseteq E \mid |I \cap A| \le c(A) \mbox{ for any } A \in \mathcal{L} \}. 
$$
By definition, partition matroids are laminar.

\paragraph{Linear matroids.}
A matroid $\M = (E, \I)$ is said to be {\em linearly representable} or a {\em linear matroid}
if there exists a matrix $A$ over some field whose columns are indexed by $E$ such that 
$I \in \I$ if and only if the columns indexed by $I$ are linearly independent. 
It is known that graphic, cographic, transversal and laminar matroids are all linearly representable.